\documentclass[english]{cccconf}
\pdfoutput=1
\usepackage[english]{babel}
\usepackage{theorem}
\theoremheaderfont{\itshape\bfseries}
{\theorembodyfont{\itshape}
	
	\newtheorem{lemma}{\textbf{Lemma}}
	\newtheorem{definition}{\textbf{Definition}}
	\newtheorem{theorem}{\textbf{Theorem}}
	
	\newtheorem{remark}{\textbf{Remark}}
	
	\newtheorem{problem}{\textbf{Problem}}
}
\usepackage{graphicx}
\usepackage{newtxtext,newtxmath}
\usepackage{amsmath}
\usepackage{amsfonts}
\usepackage{mathrsfs}
\usepackage{xcolor}
\usepackage{mdframed}
\usepackage{booktabs}
\usepackage{caption}
\usepackage{supertabular}

\newcommand{\T}{^{\mbox{\tiny T}}}
\newcommand{\R}{\mathbb{R}}

\newcommand{\eps}{\varepsilon}

\let\leq\leqslant
\let\geq\geqslant

\newenvironment{proof}[1][Proof]%
{\par\noindent\textit{#1:\ }}%
{\hspace*{\fill} \rule{6pt}{6pt}}
\newenvironment{proof*}[1][Proof]%
{\par\noindent\textit{#1:\ }}{}

\DeclareMathOperator{\diag}{diag}

\DeclareMathOperator{\sat}{sat}
\DeclareMathOperator{\sgn}{sgn}
\usepackage{tikz}
\usetikzlibrary{shapes,calc,arrows,patterns,decorations.pathmorphing
	,decorations.markings}
\usetikzlibrary{arrows.meta}

\newenvironment{system}[1]%
{\setlength{\arraycolsep}{0.5mm}\left\{ \; \begin{array}{#1}}%
	{\end{array} \right.}
\newenvironment{system*}[1]%
{\setlength{\arraycolsep}{0.5mm} \begin{array}{#1}}%
	{\end{array}}

\begin{document}

\title{\LARGE \textbf{Scalable global state synchronization of discrete-time double integrator multi-agent systems with input saturation via linear protocol}}
	\author{ Zhenwei Liu\aref{neu}, Ali Saberi\aref{wsu},
	Anton A. Stoorvogel\aref{ut}}

\affiliation[neu]{College of Information Science and
	Engineering, Northeastern University, Shenyang 110819, China
		\email{liuzhenwei@ise.neu.edu.cn}}
\affiliation[wsu]{School of Electrical Engineering and Computer
	Science, Washington State University, Pullman, WA 99164, USA
	\email{saberi@wsu.edu}}
\affiliation[ut]{Department of Electrical Engineering,
	Mathematics and Computer Science, University of Twente, Enschede, The Netherlands
	\email{A.A.Stoorvogel@utwente.nl}}

\maketitle

\begin{abstract}
   This paper studies scalable global state
   synchronization of discrete-time double integrator multi-agent systems in presence of input saturation based on
   localized information exchange. A \textit{scale-free} collaborative linear
   dynamic protocols design methodology is developed for discrete-time
   multi-agent systems with both full and partial-state couplings.
   And the protocol design methodology does not
   need any knowledge of the directed network topology and the spectrum
   of the associated Laplacian matrix. Meanwhile, the protocols are parametric based on a parameter set in which the designed protocols can guarantee the global synchronization result.
   Furthermore, the proposed protocol
   is scalable and achieves synchronization for any arbitrary number of
   agents. 
  \end{abstract}


%\footnotetext{This work is supported by Nature Science
%	Foundation of Liaoning Province under Grant 2019-MS-116.}

\section{Introduction}

In recent years, the synchronization or consensus problem of multi-agent system (MAS) has
attracted much more attention, due to its wide
potential for applications in several areas such as automotive vehicle
control, satellites/robots formation, sensor networks, and so on. See
for instance the books
\cite{bai-arcak-wen,bullobook,kocarev-book, mesbahi-egerstedt,ren-book,
	saberi-stoorvogel-zhang-sannuti,wu-book} and references therein.

At present, most work in synchronization for MAS
focused on state synchronization of continuous-time and discrete-time
homogeneous networks. State synchronization based on diffusive
\emph{full-state coupling} (it means that all states are communicated over the network) has been studied where the
agent dynamics progress from single- and double-integrator (e.g.
\cite{eichler-werner,hadjicostis-charalambous,li-zhang,saber-murray2,ren,ren-beard,tuna2}) to more general
dynamics (e.g.  \cite{tuna1,wieland-kim-allgower,you-xie}). State
synchronization based on diffusive \emph{partial-state coupling} (i.e., only part of
the states are communicated over the network) has
also been considered, including static design (\cite{chopra-tac,liu-zhang-saberi-stoorvogel-auto,liu-zhang-saberi-stoorvogel-ejc}), dynamic design
(\cite{kim-shim-back-seo,liu2018regulated,seo-shim-back,su-huang-tac,tuna3,wang-saberi-stoorvogel-grip-yang}), and the design with additional communication
(\cite{chowdhury-khalil, li-duan-chen, scardovi-sepulchre}). 

%Meanwhile,
%some works can also be found for MAS with discrete-time agents, see
%\cite{saber-murray2,ren-beard,li-zhang,hadjicostis-charalambous,%
	%	eichler-werner,tuna2} for essentially single- and double-integrator,
%and \cite{li-duan-chen,you-xie,hengster-you-lewis-xie,%
	%	lee-kim-shim,zhou-xu-duan,zhao-park-zhang-shen,wang-saberi-yang,%
	%	wang-saberi-stoorvogel-grip-yang,liu-zhang-saberi-stoorvogel-ejc,liu2018regulated}
%for higher-order agents.

On the other hand, it is worth to note that actuator saturation is pretty
common and indeed is ubiquitous in engineering applications. Some
researchers have tried to establish (semi) global state and output
synchronization results for both continuous- and discrete-time MAS in the presence of input
saturation. From the existing literature
for a linear system subject to actuator saturation, we have the following conclusion \cite{saberi-stoorvogel-zhang-sannuti}:
\begin{enumerate}
	\item A linear protocol is used if we consider synchronization in
	the semi-global framework (i.e.\ initial conditions of agents are in
	a priori given compact set).
	\item Synchronization in the global sense (i.e., when initial
	conditions of agents are anywhere) in general requires a nonlinear
	protocol.
	\item Synchronization in the presence of actuator saturation requires
	eigenvalues of agents to be in the closed left half plane for
	continuous-time systems and in the closed unit disc for
	discrete-time systems, that is the agents are at most weakly
	unstable.	
\end{enumerate}
The semi-global synchronization has been studied in \cite{su-chen-lam-lin} via
full-state coupling. For partial state coupling, we have
\cite{su-chen,zhang-chen-su}
which are based on the extra communication. Meanwhile, the result without the extra communication is developed in \cite{zhang-saberi-stoorvogel-continues-discrete}. Then, the static controllers via partial state coupling is designed in \cite{liu-saberi-stoorvogel-zhang-ijrnc} by passifying the original agent model.

On the other hand, global synchronization for full-state coupling has been studied by
\cite{meng-zhao-lin-2013} (continuous-time) and
\cite{yang-meng-dimarogonas-johansson} (discrete-time) for neutrally stable and double-integrator
agents. The global framework has only been studied for static
protocols under the assumption that the agents are neutrally stable
and the network is detailed balanced or undirected. Partial-state coupling has been
studied in \cite{chu-yuan-zhang} using an adaptive approach but the
observer requires extra communication. The result dealing with networks that are not
detailed balanced are based on \cite{li-xiang-wei} which intrinsically
requires the agents to be single integrators. Recently, we introduce a scale-free linear 
collaborative protocols for global regulated state synchronization of continuous- and discrete-time
homogeneous MAS, see \cite{liu-saberi-stoorvogel-donya-inputsaturation-automatica} and \cite{liu-saberi-stoorvogel-IJRNC-2022}. This \emph{scale-free} protocol
means the design is independent of the information about the
associated communication graph or the size of the network, i.e., the
number of agents.

In this paper, we focus on scalable linear protocol design for global state synchronization of
discrete-time double-integrator MAS in presence of input
saturation. The contributions of this paper are stated as
follows:
\begin{itemize}
%	\item We develop scalable linear protocols for a class
%	of double-integrator agent model to achieve global state synchronizations.
	\item  A class of parametric linear protocol is established based on a parameter set in which the designed parametric protocol makes all states of MAS synchronized.
	\item Meanwhile, the linear protocol design is
	scale-free and do not need any information about communication
	network. In other words, the proposed protocols work for any MAS
	with any communication graph with arbitrary number of agents as long
	as the communication graph has a path among each agent.
	%	\item A class of parametric linear protocol to scale-free in a set that the parameters make
	%	
	%	parameter solvable zone is established, and the protocol gain in this zone is proved to guarantee all states synchronized for this discrete-time MAS in presence of input
	%	saturation which consists of double-integrator dynamics.
\end{itemize}

%\cite{xie-wang-jmaa-2012} deve

\subsection*{Notations and definitions}
Given a matrix $A\in \mathbb{R}^{m\times n}$, $A\T$ denotes its
conjugate transpose and $\|A\|$ is the induced 2-norm. A square matrix
$A$ is said to be Schur stable if all its eigenvalues are in the
closed unit disk. $A\otimes B$ depicts the Kronecker product between
$A$ and $B$. $I_n$ denotes the $n$-dimensional identity matrix and
$0_n$ denotes $n\times n$ zero matrix; sometimes we drop the subscript
if the dimension is clear from the context.  A matrix
$D=[d_{ij}]_{N\times N}$ is called a row stochastic matrix if (a)
$d_{ij}>0$ for any $i,j$ and (b) $\sum_{j}^{N}d_{ij}=1$ for
$i=1,\cdots,N$.  A row stochastic matrix $D$ has at least one
eigenvalue at 1 with right eigenvector $\textbf{1}$.

A \emph{weighted graph} $\mathcal{G}$ is defined by a triple
$(\mathcal{V}, \mathcal{E}, \mathcal{A})$ where
$\mathcal{V}=\{1,\ldots, N\}$ is a node set, $\mathcal{E}$ is a set of
pairs of nodes indicating connections among nodes, and
$\mathcal{A}=[a_{ij}]\in \mathbb{R}^{N\times N}$ is the weighting
matrix. Each pair in $\mathcal{E}$ is called an \emph{edge}, where
$a_{ij}>0$ denotes an edge $(j,i)\in \mathcal{E}$ from node $j$ to
node $i$ with weight $a_{ij}$. Moreover, $a_{ij}=0$ if there is no
edge from node $j$ to node $i$. We assume there are no self-loops,
i.e.\ we have $a_{ii}=0$. A \emph{path} from node $i_1$ to $i_k$ is a
sequence of nodes $\{i_1,\ldots, i_k\}$ such that
$(i_j, i_{j+1})\in \mathcal{E}$ for $j=1,\ldots, k-1$. A
\emph{directed tree} with root $r$ is a subgraph of the graph
$\mathcal{G}$ in which there exists a unique path from node $r$ to
each node in this subgraph. A \emph{directed spanning tree} is a
directed tree containing all the nodes of the graph. A directed graph may contain many directed spanning trees, and thus there may
be several choices for the root agent. The set of all possible root agents for a graph
$\mathcal{G}$ is denoted by $\pi_g$.

The
\emph{weighted in-degree} of node $i$ is given by 
\[d_{\text{in}}(i) = \sum_{j=1}^N\, a_{ij}.\]
For a weighted graph $\mathcal{G}$, the matrix
$L=[\ell_{ij}]$ with
\[
\ell_{ij}=
\begin{system}{cl}
	\sum_{k=1}^{N} a_{ik}, & i=j,\\
	-a_{ij}, & i\neq j,
\end{system}
\]
is called the \emph{Laplacian matrix} associated with the graph
$\mathcal{G}$. The Laplacian matrix $L$ has all its eigenvalues in the
closed right half plane and at least one eigenvalue at zero associated
with right eigenvector $\textbf{1}$ \cite{royle-godsil}.

\section{Problem formulation}

Consider a MAS consisting of $N$ identical discrete-time double integrator with input saturation:
\begin{equation}\label{eq1}
	\begin{cases}
		{x}_i(k+1)=Ax_i(k)+B\sigma(u_i(k)),\\
		y_i(k)=Cx_i(k)
	\end{cases}
\end{equation}
where $x_i(k)\in\mathbb{R}^{2n}$, $y_i(k)\in\mathbb{R}^n$ and
$u_i(k)\in\mathbb{R}^n$ are the state, output, and the input of agent 
$i=1,\ldots, N$, respectively. And
\[
A=\begin{pmatrix}
	I&I\\0&I
\end{pmatrix}, B=\begin{pmatrix}0\\I\end{pmatrix},C=\begin{pmatrix}
	I&0
\end{pmatrix}
\]
Meanwhile,
\begin{equation*}%\label{satdef}
	\sigma(v)=\begin{pmatrix} 
		\sat(v_1) \\ \sat(v_2) \\ \vdots \\ \sat(v_m)
	\end{pmatrix}\quad\text{ where }\quad
	v=\begin{pmatrix} 
		v_1 \\ v_2 \\ \vdots \\ v_m
	\end{pmatrix} \in \R^m
\end{equation*}
with $\sat(w)$ is the standard saturation function,
\[
\sat(w)=\sgn(w)\min(1,|w|).
\]

The network provides agent $i$ with the following information,
\begin{equation}\label{eq2}
	\zeta_i(k)=\sum_{j=1}^{N}a_{ij}(y_i(k)-y_j(k)),
\end{equation}
where $a_{ij}\geq 0$ and $a_{ii}=0$. This communication topology of
the network can be described by a weighted graph $\mathcal{G}$ associated with \eqref{eq2}, with
the $a_{ij}$ being the coefficients of the weighting matrix
$\mathcal{A}$. In terms of the coefficients of the associated
Laplacian matrix $L$, $\zeta_i$ can be rewritten as
\begin{equation}\label{zeta_l}
	\zeta_i(k) = \sum_{j=1}^{N}\ell_{ij}y_j(k).
\end{equation}
We refer to this as \emph{partial-state coupling} since only part of
the states are communicated over the network. When $C=I$, it means all states are communicated over the network, we call it \emph{full-state coupling}. Then, the original agents are expressed as
\begin{equation}\label{neq1}
	{x}_i(k+1)=Ax_i(k)+B\sigma (u_i(k))
\end{equation}
%and
%\begin{equation}\label{nsolu-cond}
%\dot{x}_r  = A x_r
%\end{equation}
and $\zeta_i$ is rewritten as
\begin{equation}\label{zeta-f}
	\zeta_i(k)= \sum_{j=1}^{N}\ell_{ij}x_j(k).
\end{equation}

%\begin{assumption}\label{Aass}
%	We assume that:
%	\begin{itemize}		
	%		\item[(i)] $(A,B)$ are stabilizable and $(C,A)$ are detectable.
	%		\item[(ii)] All eigenvalues of $A$ are in the closed left half plane. 
	%	\end{itemize}
%\end{assumption}

%In the following sections, we provide two types of scale-free 

%For communication network, we have the following assumption:
%\begin{assumption}\label{ass2}
%	The graph $\mathcal{G}^N$ describes the communication topology of the network containing a directed spanning tree.
%\end{assumption}

We need the following definition to explicitly state our problem
formulation.
\begin{definition}\label{def1}  
	We define the following set. $\mathbb{G}^N$ denotes the set of directed graphs of $N$ agents
	which contains a directed spanning tree. Moreover, for any $\mathcal{G}\in\mathbb{G}^N$, we denote the root set of the $\mathcal{G}$ by $\pi_g$.
\end{definition}

\begin{remark}
	When the undirected or strongly connected graph is considered, it is obvious that the set $\pi_g$ will include all nodes of networks. 
\end{remark}

We consider the \textbf{state synchronization} problem under the graph set $\mathbb{G}^N$ satisfying Definition \ref{def1}. Here, its objective is that the agents achieve state
synchronization, that is
\begin{equation}\label{synch_ss}
	\lim_{k\to \infty} (x_{i}(k)-x_j(k))=0.
\end{equation}
for all $i,j \in {1,...,N}$.

Meanwhile, we introduce an additional
information exchange among each agent and its neighbors. In particular, each agent 
$i=1,\ldots, N$ has access to additional information, denoted by
$\hat{\zeta}_i$, of the form
\begin{equation}\label{eqa1}
	\hat{\zeta}_i(k)=\sum_{j=1}^Na_{ij}(\xi_i(k)-\xi_j(k))
\end{equation}
where $\xi_j\in\mathbb{R}^n$ is a variable produced internally by agent $j$ and to be defined in next sections.

%In this paper, we introduce an additional information exchange among
%protocols. In particular, each agent $i=1,\ldots, N$ has access to
%additional information, denoted by $\hat{\zeta}_i(k)$, of the form
%\begin{equation}\label{eqa2}
%	\hat{\zeta}_i(k)= \sum_{j=1}^N{a}_{ij}(\xi_i(k)-\xi_j(k)) 
%\end{equation}
%where $\xi_j(k)\in\mathbb{R}^n$ is a signal produced internally by
%agent $j$ which will be selected as part of the protocol design in the
%next sections.

Then, we formulate the problem for global state synchronization of a MAS via
linear protocols based on additional information exchange \eqref{eqa1}.
\begin{problem}\label{prob4ss}
	Consider a MAS described by \eqref{eq1} and \eqref{eq2}. Let the set $\mathbb{G}^N$ denote all graphs satisfy Definition \ref{def1}.
	
	The \textbf{scalable global state synchronization problem
		with additional information exchange via linear dynamic protocol} is to find a linear dynamic protocol,
	using only the knowledge of agent model $(A,B,C)$, of the form 
	\begin{equation}\label{protoco5ss}
		\begin{system}{cl}
			x_{c,i}(k+1)=&A_{c,i} x_{c,i}(k)+B_{c,i}{\sigma(u_i(k))}\\
			&\hspace{1.5cm}+C_{c,i}
			{\zeta}_i(k)+D_{c,i} \hat{\zeta}_i(k),\\
			u_i(k)=&K_{c,i}x_{c,i}(k)
		\end{system}
	\end{equation}
	where $\hat{\zeta}_i$ is defined in \eqref{eqa1} with
	$\xi_i=H_{c,i}x_{c,i}$, and $x_{c,i}\in\R^{n_c}$, such that 
	state synchronization \eqref{synch_ss} is achieved for any $N$ and any
	graph $\mathcal{G}\in \mathbb{G}^N$, and for all
	initial conditions of the agents $x_i(0) \in \mathbb{R}^n$, and
	all initial conditions of the protocols
	$x_{c,i}(0) \in \mathbb{R}^{n_c}$.
\end{problem}

\section{Protocol design}

\subsection{Full-state coupling}
Let $\mathcal{G}$ be any graph belongs to $\mathbb{G}^N$, and we choose agent $\theta$ where $\theta$ is any node in the root set $\pi_g$. Then, we propose the following protocol.

%\begin{center}
%	%\textbf{Protocol 1:} Full-state coupling
%	\textbf{Linear Protocol 1:}  Full-state coupling \vspace*{-3mm}
%\end{center}
\begin{tabular}{p{6.6cm}}
	\toprule
	\textbf{Linear Protocol 1:}  Full-state coupling\\ %\vspace*{-3mm}
	\toprule
	\vspace*{-7mm}	
	\begin{equation}\label{pscpd1ss}
		\begin{system}{cll}
			{\chi}_i(k+1) &=&
			A\chi_i(k)+B\sigma(u_i(k))\\
			&&\quad+ \frac{1}{1+D_{in}(i)}
			\left[ A{\zeta}_i(k)-A\hat{\zeta}_i(k)
			\right] \\
			u_i(k) &=&K\chi_i(k), \quad i=\{1,\ldots, N\}\setminus\theta\\
			u_\theta(k)&\equiv &0, 
		\end{system}	
	\end{equation}
	where $D_{\text{in}}(i)$ is the upper bound of $d_{in}(i)=\sum_{j=1}^N a_{ij}$. Then, we still choose matrix $K=-\begin{pmatrix}
		k_1 I&k_2 I
	\end{pmatrix}$, where $k_1\in(0, 1)$ and  $k_2 >0$ satisfy the following condition 
\begin{equation}\label{cond1}
(1+k_1-k_2)^2< 1-k_1.
\end{equation}
	$\hat{\zeta}_i(k)$ and ${\zeta}_i(k)$ are defined
	by \eqref{eqa1} and \eqref{eq2}, respectively. And the agents communicate $\xi_i(k)$ which is chosen as $\xi_i(k)=\chi_i(k)$. \\
	\bottomrule
\end{tabular}

\begin{remark}
	$D_{\text{in}}(i)$ is an \emph{upper bound} for the \emph{weighted
		in-degree} $d_{\text{in}}(i)=\sum_{j=1}^N a_{ij}$ for node $i$. 
	It is still local information. In our protocol design, the bound
	$D_{\text{in}}(i)$ is used to scale the communication among agents, which can otherwise cancel
	the impact of the $\theta$th weighting value. 
\end{remark}
The condition \eqref{cond1} can be shown as Fig.~\ref{zone}, where the zone encircled by parabola and line $(0,0)$ to $(0,2)$.
\begin{figure}[ht!]
	\includegraphics[width=5cm, height=6.5cm]{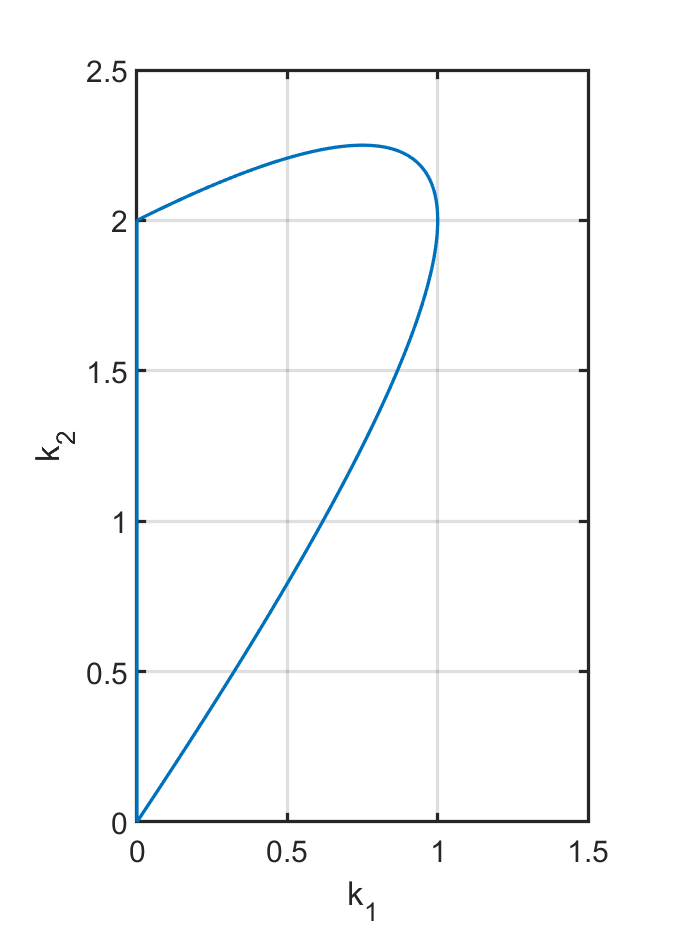}
	\centering
	\caption{Solvable zone of $k_1, k_2$ for synchronization}\label{zone}
\end{figure}

\begin{theorem}\label{mainthm1}
	Consider a MAS described by \eqref{neq1} and \eqref{zeta-f}. Let the set $\mathbb{G}^N$ denote all graphs satisfy Definition \ref{def1}.
	
	Then, the scalable global state synchronization problem with additional information exchange as stated in Problem
	\ref{prob4ss} is solvable. In particular, for any given $k_1\in(0,1)$ and $k_2 >0$ satisfying \eqref{cond1}, the linear
	dynamic protocol \eqref{pscpd1ss} solves the global state
	synchronization problem for any $N$ and any graph
	$\mathcal{G}\in\mathbb{G}^N$. 
\end{theorem}

To obtain this theorem we need the following lemma.
\begin{lemma}\label{lem1} For all $u,v\in \R^n$, we have
	\begin{equation}\label{cond0}
		(\sigma(v)-\sigma(u))\T(u-\sigma(u))\leq 0.
	\end{equation}
\end{lemma}
\begin{proof}
	Note that we have:
	\begin{multline}\label{star}
		(\sigma(v)-\sigma(u))\T(u-\sigma(u)) \\ =\sum_{i=1}^n
		(\sigma(v_i)-\sigma(u_i))(u_i-\sigma(u_i))
	\end{multline}
	when:
	\[
	u=\begin{pmatrix} u_1 \\ \vdots \\ u_n \end{pmatrix},\qquad
	v=\begin{pmatrix} v_1 \\ \vdots \\ v_n \end{pmatrix}
	\]
	Next note that if $u_i\geq 1$ we have
	$\sigma(v_i)-\sigma(u_i)=\sigma(v_i)-1\leq0$ and
	$u_i-\sigma(u_i)=u_i-1\geq 0$ and hence:
	\begin{equation}\label{starstar}
		(\sigma(v_i)-\sigma(u_i))(u_i-\sigma(u_i))\leq 0
	\end{equation}
	On the other hand if $u_i\leq -1$ we have
	$\sigma(v_i)-\sigma(u_i)=\sigma(v_i)+1\geq0$ and
	$u_i-\sigma(u_i)=u_i+1\leq 0$ and \eqref{starstar} is still
	satisfied.
	Finally, if  $|u_i|\leq 1$ then $u_i-\sigma(u_i)=0$ and
	\eqref{starstar} is also satisfied.
	
	Since \eqref{starstar} is satisfied for all $i$ and using
	\eqref{star} we find \eqref{cond0} holds for all $u$ and $v$.
\end{proof}
\vspace{1cm}
\begin{proof}[The proof of Theorem \ref{mainthm1}]
	Since we have $u_\theta(k)\equiv0$, we obtain $\sigma(u_\theta(k))=0$. The model of agent $\theta$ is rewritten as
	\[
	x_\theta(k+1)=Ax_\theta(k)
	\]
	
	Then, let $\bar{x}(k)=x_i(k)-x_\theta(k)$, we have
	\begin{equation}
		\begin{system}{l}
			\bar{x}_{i}(k+1)=A\bar{x}_{i}(k)+B\sigma(u_i(k))\\
			\chi_i(k+1)=A\chi_i(k)+B\sigma(u_i(k))\\
			\qquad\qquad\quad+ \frac{1}{1+D_{{in}}(i)}
			\sum_{j=1}^{N-1}{\ell}_{ij} A\left[\bar{x}_{i}(k)-\chi_i(k) \right] \\
			u_i(k)=-\begin{pmatrix}k_1 I&k_2  I\end{pmatrix}\chi_i(k)
		\end{system}
	\end{equation}
	Then by defining $2(N-1)n$-dimensional vectors
	\begin{multline*}
		\bar{x}(k)=\begin{pmatrix}
			\bar{x}_1(k)\\\vdots\\\bar{x}_{N-1}(k)
		\end{pmatrix},   
		{\chi}(k)=\begin{pmatrix}
			\chi_1(k)\\\vdots\\\chi_{N}(k)
		\end{pmatrix}, \\u(k)=\begin{pmatrix}
			u_1(k)\\\vdots\\u_N(k)
		\end{pmatrix},\sigma(u(k))=\begin{pmatrix}
			\sigma(u_1(k))\\\vdots\\\sigma(u_N(k))
		\end{pmatrix}
	\end{multline*}
	where $\chi_\theta(k)$, $u_\theta(k)$, and $\sigma(u_\theta(k))$ are not included. We have the following closed-loop system
	\begin{equation}
		\begin{system}{l}
			\bar{x}(k+1)=(I_{N-1}\otimes A)\bar{x}(k)+(I_{N-1}\otimes B)\sigma(u(k))\\
			\chi(k+1)=(I_{N-1}\otimes A)\chi(k)+(I_{N-1}\otimes B)\sigma(u(k))\\
			\hspace{2cm}+((I_{N-1}-\bar{D})\otimes A)(\bar{x}(k)-\chi(k))\\
			u(k)=-(I_{N-1}\otimes \begin{pmatrix}k_1 I&k_2 I\end{pmatrix})\chi(k)
		\end{system}
	\end{equation}
	where $\bar{D}=I_{N-1}-(I_{N-1}+D_{d,\text{in}})^{-1}\hat{L}$,
	\[
	D_{d,\text{in}}=\diag\{D_{in}(1),D_{in}(2),\cdots,D_{in}(N)\} \text{ without } D_{in}(\theta),
	\]
	and $\hat{L}$ is the matrix obtained from $L$ by
	deleting the $\theta$th row and the $\theta$th column. Meanwhile, according to \cite[Lemma 1]{grip-yang-saberi-stoorvogel-automatica}, we have the real part of all eigenvalues of $\hat{L}$ are greater than zero. Thus, it implies all eigenvalues' absolute value of
	$\bar{D}\in\R^{(N-1)\times (N-1)}$ are less than 1.
	
	Let $e(k)=\bar{x}(k)-\chi(k)$, we have
	\begin{equation}\label{eqn1}
		\begin{system}{l}
			\bar{x}(k+1)=(I_{N-1}\otimes A)\bar{x}(k)+(I_{N-1}\otimes B)\sigma(u(k))\\
			e(k+1)=(\bar{D}\otimes A)e(k)\\
			u(k)=-(I_{N-1}\otimes \begin{pmatrix}k_1 I&k_2 I\end{pmatrix})(\bar{x}(k)-e(k))
		\end{system}
	\end{equation}	
	Then, let
	\[
	\bar{x}_i(k)=\begin{pmatrix}
		\bar{x}_{1,i}(k)\\
		\bar{x}_{2,i}(k)
	\end{pmatrix},
	\]
	we have
	\begin{equation}\label{eqo1}
		\begin{system}{l}
			\bar{x}_{1}(k+1)=\bar{x}_{1}(k)+\bar{x}_{2}(k)\\
			\bar{x}_{2}(k+1)=\bar{x}_{2}(k)+\sigma(u(k))\\
			e(k+1)=(\bar{D}\otimes A)e(k)\\
			u(k)=-k_1 \bar{x}_{1}(k)-k_2 \bar{x}_{2}(k)+\left(I_{N-1}\otimes \begin{pmatrix}k_1 I&k_2 I\end{pmatrix}\right)e(k)
		\end{system}
	\end{equation}

	%	Here we have that the eigenvalues of
	%$\bar{D}\in\R^{N\times N}$ defined by \eqref{hodt-LDa} are less than 1.
	
	%	Next, we transform system \eqref{eqo1} to prove its stability. Let $\eta_1(k)=u(k)$ and $\eta_2(k)=k_1 \bar{x}_2(k)$, we have
	%	\begin{equation}\label{eqt1}
		%		\begin{system}{l}
			%			\eta_1(k+1)=\eta_1(k)-\eta_2(k)-k_2 \sigma(\eta_1(k))\\
			%			\hspace{2cm}+[I_{N-1}\otimes \begin{pmatrix}
				%				k_1 I&k_2 I
				%			\end{pmatrix}] (\bar{D}\otimes A-I)e(k)\\
			%			\eta_2(k+1)=\eta_2(k)+k_1 \sigma(\eta_1(k))\\
			%			e(k+1)=(\bar{D}\otimes A)e(k)
			%		\end{system}
		%	\end{equation}

	The eigenvalues of $\bar{D}\otimes A$ are of the form
	$\lambda_i \mu_j$, with $\lambda_i$ and $\mu_j$ eigenvalues of
	$\bar{D}$ and $A$, respectively. Since $|\lambda_i|<1$ and
	$\mu_j\equiv 1$, we find $\bar{D}\otimes A$ is asymptotically
	stable. Therefore we find that:
	\begin{equation}\label{estable}
		\lim_{k\to \infty}e_i(k)\to 0.
	\end{equation}
	It also shows that $e_i\in \ell_2$.  Thus, we just need to prove
	the stability of \eqref{eqo1}.
	%	i.e., 
	%	\begin{equation}\label{eqt2}
		%		\begin{system}{l}
			%			\eta_1(k+1)=\eta_1(k)-\eta_2(k)-k_2 \sigma(\eta_1(k))\\
			%			\eta_2(k+1)=\eta_2(k)+k_1 \sigma(\eta_1(k))
			%		\end{system}
		%	\end{equation}
	%Meanwhile, if $e(k), \eta_2(k)\to 0$, we obtain 
	%$\eta_1(k)\to 0$. 
	Thus, we have $\bar{x}(k)\to 0$ as $k\to \infty$ with $e_i\in \ell_2$, which will obtain the synchronization result.
	
	%Then, according to the condition \eqref{cond1}, we can obtain the following solvable zone (the triangular zone) as Fig \ref{zone} shown.

	%	To prove the stability of system \eqref{eqt1} for any $k_1$ and $k_2$ in the above zone, we split the zone to two parts, Zone I and II as Fig \ref{zone} shown. That means
	%	\[
	%	\begin{system*}{l}
		%		k_1, k_2 \text{ belong to Zone I,  if }(1+k_1-k_2)^2< 1-k_1\\
		%		k_1, k_2 \text{ belong to Zone II, if }(1+k_1-k_2)^2\geq 1-k_1
		%	\end{system*}
	%	\] 
	
	To prove the synchronization result, we consider the following weighting Lyapunov function
	\begin{equation}\label{lyapunvd1}
		V(k)=(1-h)V_1(k)+hV_2(k)
	\end{equation}
	where $h\in(0,1)$,
	\begin{align*}
		V_1(k)=&\begin{pmatrix}
			\sigma(u(k))\\\bar{x}_2(k)
		\end{pmatrix}\T
		\left[\begin{pmatrix}
			1& k_1\\ k_1&{k_1}
		\end{pmatrix}\otimes I_{(N-1)n}\right]
		\begin{pmatrix}
			\sigma(u(k))\\\bar{x}_2(k)
		\end{pmatrix}\\
		&\hspace{2cm}+2\sigma(u(k))\T( u(k)-\sigma(u(k)))\\
		V_2(k)=&e\T(k) P_De(k)
	\end{align*}
	and $P_D>0$ satisfies 
	\begin{equation}\label{condd3}
		(\bar{D}\otimes A)\T P_D(\bar{D}\otimes A)- P_D\leq -2I_{2(N-1)n}.
	\end{equation}
	Here, we obtain $V_1(k)$ and $V_2(k)$ are positive, i.e. $V_1(k)>0$ except for $(u(k), \bar{x}_2(k))=0$ when $V_1(k)=0$ and $V_2(k)>0$ except for $e(k)=0$ when $V_2(k)=0$. 
	Then, we have
	\begin{align*}
		&\Delta V_1(k)=V_1(k+1)-V_1(k)\\
		=&-\sigma(u(k+1))\T \sigma(u(k+1))+2\sigma(u(k+1))\T u(k)\\
		&+2(k_1-k_2) \sigma(u(k+1))\T\sigma(u(k))\\		
		&+(1+k_1) \sigma(u(k))\T\sigma(u(k))-2\sigma(u(k))\T u(k)\\
		&+2\sigma(u(k+1))\T(I_{N-1}\otimes (k_1 I \quad k_2 I)\Psi)e(k)\\
		=&2(\sigma(u(k+1))-\sigma(u(k)))\T(u(k)-\sigma(u(k)))\\
		&+2(1+k_1-k_2) \sigma(u(k+1))\T\sigma(u(k))\\
		&-\sigma(u(k+1))\T \sigma(u(k+1))-(1-k_1)\sigma(u(k))\T\sigma(u(k))\\
		&+2\sigma(u(k+1))\T(I_{N-1}\otimes (k_1 I \quad k_2 I)\Psi)e(k)\\
		\leq&2(1+k_1-k_2) \sigma(u(k+1))\T\sigma(u(k))\\
		%	\end{align*}
	%	\begin{align*}				
		&-\sigma(u(k+1))\T \sigma(u(k+1))-(1-k_1)\sigma(u(k))\T\sigma(u(k))\\
		&+2\sigma(u(k+1))\T(I_{N-1}\otimes (k_1 I \quad k_2 I)\Psi)e(k)%		
	\end{align*}
	since $(\sigma(u(k+1))-\sigma(u(k)))\T(u(k)-\sigma(u(k)))\leq 0$ based on Lemma \ref{lem1}, where $\Psi=\bar{D}\otimes A-I_{2(N-1)n}$.
	Meanwhile, for $V_2(k)$ we have
	\begin{equation*}
		\Delta V_2(k)=V_2(k+1)-V_2(k)\leq -2e\T(k) e(k)
	\end{equation*}
	based on condition \eqref{condd3}.	
	Thus, one can obtain
	\begin{align*}
		&\Delta V(k)\leq (1-h)\Delta V_1(k)+h\Delta V_2(k)\\
		\leq &2(1-h)(1+k_1-k_2) \sigma(u(k+1))\T\sigma(u(k))\\
		&-(1-h)\left(1-\frac{\|\Psi\|^2(1-h)(k_1^2+k_2^2)}{h}\right)\\
		&\times\|\sigma(u(k+1))\|^2-(1-h)(1-k_1)\sigma(u(k))\T\sigma(u(k))\\
		&-he\T(k) e(k)\\
		=&(1-h)\begin{pmatrix}
			\sigma(u(k+1))\\\sigma(u(k))
		\end{pmatrix}\T (\Phi\otimes I_{(N-1)n} ) \begin{pmatrix}
			\sigma(u(k+1))
			\\\sigma(u(k))
		\end{pmatrix}\\
	&-he\T(k) e(k)
	\end{align*}	
	where
	\begin{equation}\label{cond2}
		\Phi=\begin{pmatrix}
			-1+\frac{\|\Psi\|^2(1-h)(k_1^2+k_2^2)}{h}&1+k_1-k_2\\
			1+k_1-k_2&-(1-k_1)
		\end{pmatrix}.
	\end{equation}
	
	Obviously we just need to prove $\Phi< 0$. Without loss of generality, there exists an $\eps>0$ such that 
	\begin{equation}\label{cond3}
		\frac{(1+k_1-k_2)^2}{1-k_1}= 1-\eps.
	\end{equation}
	By using Schur Compliment, we have $\Phi< 0$ is equivalent to
	\[
	-1+\frac{\|\Psi\|^2(1-h)(k_1^2+k_2^2)}{h}+\frac{(1+k_1-k_2)^2}{1-k_1}< 0.
	\]
	From condition \eqref{cond3}, we can obtain
	\begin{align*}
		&-1+\frac{\|\Psi\|^2(1-h)(k_1^2+k_2^2)}{h}+\frac{(1+k_1-k_2)^2}{1-k_1}\\
		< &\frac{\|\Psi\|^2(1-h)(k_1^2+k_2^2)}{h}-\eps
	\end{align*}
    For $h$ sufficiently close to 1, one can obtain
	\[
	\frac{\|\Psi\|^2(1-h)(k_1^2+k_2^2)}{h}<\eps
	\]
	It means that we obtain  $\Phi< 0$.
	
	Thus, we have $\Delta V(k)< 0$ for $\begin{pmatrix}
		\sigma(u(k+1))
		\\\sigma(u(k))
	\end{pmatrix}\neq 0$, $\bar{x}(k)\to 0$ as $k\to \infty$. 

Furthermore, when $\Delta V(k)=0$, we obtain $u(k+1)=u(k)=0$ and $e(k)=0$. It is easy to obtain $\bar{x}_1(k)=\bar{x}_2(k)=0$ at $\Delta V(k)=0$. 

Thus, the invariance set
$\{(\bar{x}(k),e(k)): \Delta{V}(\bar{x}(k),e(k))=0\}$ contains
no trajectory of the system except the trivial trajectory
$(\bar{x}(k),e(k))=(0,0)$. \eqref{eqn1} is globally asymptotically stable based on LaSalle’s invariance principle. 

Finally, we obtain the global state synchronization result.	
\end{proof}

\begin{remark}
	From Theorem \ref{mainthm1}, we can know that the solvable zone shown in Fig. \ref{zone} is an open zone, i.e., $\frac{(1+k_1-k_2)^2}{1-k_1}< 1$ and $k_1>0$. However, if let $k_1=1$ and $k_2=2$, then we just need $\frac{\|\Psi\|^2(1-h)(k_1^2+k_2^2)}{h}<1$ to guarantee $\Phi\leq 0$. Obviously, it is obtained for some $h\to 1$. Thus, the point $(1, 2)$ is also a pair of solvable parameter for scalable synchronization result.
\end{remark}

\subsubsection{Partial-state coupling}
Let $\mathcal{G}$ be any graph belongs to $\mathbb{G}^N$, and also we choose agent $\theta$ where $\theta$ is any node in the root set $\pi_g$. Then, we propose the following linear protocol.

%\begin{center}
%	\textbf{Linear protocol 2:} Partial-state coupling\vspace*{-3mm}
%\end{center}
\begin{tabular}{p{6.6cm}}
	\toprule
	\textbf{Linear protocol 2:} Partial-state coupling\\
	\toprule
	\vspace*{-7mm}
	\begin{equation}\label{pscpd2ss}
		\begin{system}{cll}
			\hat{x}_i(k+1) &=&
			(A-FC)\hat{x}_i(k)\\
			&&\quad+\frac{1}{1+D_{{in}}(i)} \left[ B\hat{\zeta}_{i2}(k)+F{\zeta}_i(k)\right] \\ 
			{\chi}_i(k+1) &=& A\chi_i(k)+Bu_i(k)\\
			&&\quad+A\hat{x}_i(k) -\frac{1}{1+D_{{in}}(i)}A\hat{\zeta}_{i1}(k)\\
			u_i(k) &=&  K \chi_i(k), \quad i=\{1,\ldots, N\}\setminus\theta\\
			u_\theta(k)&\equiv &0, 
		\end{system}
	\end{equation}
	where $D_{\text{in}}(i)$ is the upper bound of $d_{in}(i)=\sum_{j=1}^N a_{ij}$. Then, we choose matrix $K=-\begin{pmatrix}
		k_1 I&k_2 I
	\end{pmatrix}$, where $k_1\in(0,1),k_2>0$ satisfy condition \eqref{cond1}. In this
	protocol, the agents communicate
	\[
	\xi_i(k)=\begin{pmatrix} \xi_{i1}(k) \\ \xi_{i2}(k) 
	\end{pmatrix}=\begin{pmatrix}
		\chi_i(k) \\ \sigma (u_i(k))
	\end{pmatrix},
	\]
	i.e.\ each agent has access to additional
	information
	\[
	\hat{\zeta}_i(k)=\begin{pmatrix} 
		\hat{\zeta}_{i1}(k) \\ \hat{\zeta}_{i2} (k)
	\end{pmatrix},
	\]
	where:
	\begin{equation}\label{add_1}
		\begin{system}{l}
			\hat{\zeta}_{i1}(k)=\sum_{j=1}^N a_{ij}(\chi_i(k)-\chi_j(k)) \\
			\hat{\zeta}_{i2}(k)=\sum_{j=1}^{N} a_{ij}(\sigma (u_i(k))-\sigma (u_j(k)))
		\end{system}		
	\end{equation}
	while $\bar{\zeta}_i(k)$ is defined via \eqref{eq2}.\\
	\bottomrule
\end{tabular}

We have following theorem.
\begin{theorem}\label{mainthm2}
	Consider a MAS described by \eqref{eq1} and \eqref{eq2}. Let the set $\mathbb{G}^N$ denote all graphs satisfy Definition \ref{def1}.  
	
	Then, the scalable global state synchronization problem with additional information exchange as stated in Problem
	\ref{prob4ss} is solvable. In particular, for any given $k_1\in(0,1)$ and $k_2>0$ satisfying \eqref{cond1}, the linear
	dynamic protocol \eqref{pscpd2ss} solves the global state
	synchronization problem for any $N$ and any graph
	$\mathcal{G}\in\mathbb{G}^N$. 
\end{theorem}

\begin{proof}[The proof of Theorem \ref{mainthm2}]
	Similar to Theorem \ref{mainthm1}, by defining
	$\bar{x}_i(k)=x_i(k)-x_\theta (k)$, $e(k)=\bar{x}(k)-\chi(k)$, and
	$\bar{e}(k)=[(I_{N-1}-\bar{D})\otimes I]\bar{x}(k)-\hat{x}(k)$, we have
	the matrix expression of closed-loop system
	\begin{equation*}\label{newsysted}
		\begin{system}{l}
			\bar{x}_{1}(k+1)=\bar{x}_{1}(k)+\bar{x}_{2}(k)\\
			\bar{x}_{2}(k+1)=\bar{x}_{2}(k)+\sigma(u(k))\\
			e(k+1)=(\bar{D}\otimes A)e(k)+\bar{e}(k)\\
			\bar{e}(k+1)=[I_{N-1}\otimes (A-FC)]\bar{e}(k)\\
			u(k)=-\left(I_{N-1}\otimes \begin{pmatrix}k_1 I&k_2 I\end{pmatrix}\right)\chi(k)
		\end{system}
	\end{equation*}  
	
	%	Similarly, we have the following system by let $u(k)=u(k)$ and $\eta_2(k)=k_1 \bar{x}_2(k)$,
	%	\begin{equation*}\label{newsysted2}
		%		\begin{system*}{l}
			%			\eta_{1}(k+1)=\eta_{1}(k)-\eta_{2}(k)-k_2 \sigma(u(k))\\
			%			\hspace{2cm}+I_{N-1}\otimes \begin{pmatrix}
				%				k_1 I&k_2 I
				%			\end{pmatrix} (\bar{D}\otimes A-I)e(k)\\
			%			\eta_{2}(k+1)=\eta_{2}(k)+k_1\sigma(u(k))\\
			%			e(k+1)=(\bar{D}\otimes A)e(k)+\bar{e}(k)\\
			%			\bar{e}(k+1)=[I_{N-1} \otimes (A-FC)]\bar{e}(k)
			%		\end{system*}
		%	\end{equation*} 
	Since the eigenvalues of $A-FC$ and $\bar{D}\otimes A$ are in open unit disk, we just need to prove 
	the stability of $\bar{x}_{1}(k)$ and $\bar{x}_{2}(k)$.
	%	\begin{equation*}
		%		\begin{system}{l}
			%			u(k+1)=u(k)-\eta_2(k)-k_2 \sigma(u(k))\\
			%			\eta_2(k+1)=\eta_2(k)+k_1 \sigma(u(k))
			%		\end{system}
		%	\end{equation*}
	
	Similar to the proof of Theorem \ref{mainthm1}, the state synchronization
	result can be obtained.
\end{proof}

\section{Numerical examples}
In this section, we will illustrate the effectiveness and scalability of our designs for discrete time double-integrator MAS 
by one numerical example. We use agent models \eqref{eq1} with parameters 
\[
A=\begin{pmatrix}
	1&1\\0&1
\end{pmatrix}, B=\begin{pmatrix}0\\1\end{pmatrix},C=\begin{pmatrix}
	1&0
\end{pmatrix}
\]

Meanwhile, we use three graphs which consists of 4, 7, and 60 agents respectively. We consider the case of global state synchronization with
partial-state coupling. To show the effectiveness of our protocol design based on condition \eqref{cond1}, we choose $F=[1.5\quad 0.5]\T$, $k_1=0.5$ and $k_2=1$. Furthermore, we choose $\theta=1$.
The protocol is provided as follows:
%\begin{tcolorbox}[colback=white,breakable]
\begin{equation}\label{protoclsim}
	\begin{system}{ll}
		\hat{x}_i(k+1) &=
		\begin{pmatrix}
			-0.5&1\\-0.5&1
		\end{pmatrix}\hat{x}_i(k)
		\\
		&\qquad +\frac{1}{1+D_{\text{in}}(i)} \left[  \begin{pmatrix}
			0\\1
		\end{pmatrix}\hat{\zeta}_{i2}(k)+\begin{pmatrix}
			1.5\\0.5
		\end{pmatrix}{\zeta}_i(k)\right] \\ 
		{\chi}_i(k+1) &= \begin{pmatrix}
			1&1\\0&1
		\end{pmatrix}\left[\chi_i(k)-\frac{1}{1+D_{\text{in}}(i)}
		 \hat{\zeta}_{i1}(k) \right]\\
		&\qquad+\begin{pmatrix}
			0\\1
		\end{pmatrix}\sigma(u_i(k))+\begin{pmatrix}
			1&1\\0&1
		\end{pmatrix}\hat{x}_i(k)\\
		u_i(k) &=  -\begin{pmatrix}
			k_1 &k_2 
		\end{pmatrix} \chi_i(k).
	\end{system}
\end{equation}
with $i=2,\cdots,N$.
%\end{tcolorbox}

Now we are creating three homogeneous MAS with different number of
agents and different communication topologies to show that the designed
protocol is scale-free, independent of the communication network, and number of
agents $N$. 

%\begin{figure}[ht!]
%	\includegraphics[width=5.5cm, height=3.2cm]{graph_3}
%	\centering
%	\caption{Communication graph for Case I}\label{graph_3}
%\end{figure}

\subsection*{Case I: 4-agent graph}
In this case, we consider a MAS with $4$ agents,
$N=4$. The associated adjacency matrix to the communication network is assumed to be $\mathcal{A}_I$ where $a_{21}=a_{32}=a_{43}=1$. 

The simulation results of Case I by using protocol \eqref{protoclsim} are demonstrated in Figure
\ref{results_case1.1}. 
\begin{figure}[ht!]
	\includegraphics[width=8.5cm, height=5cm]{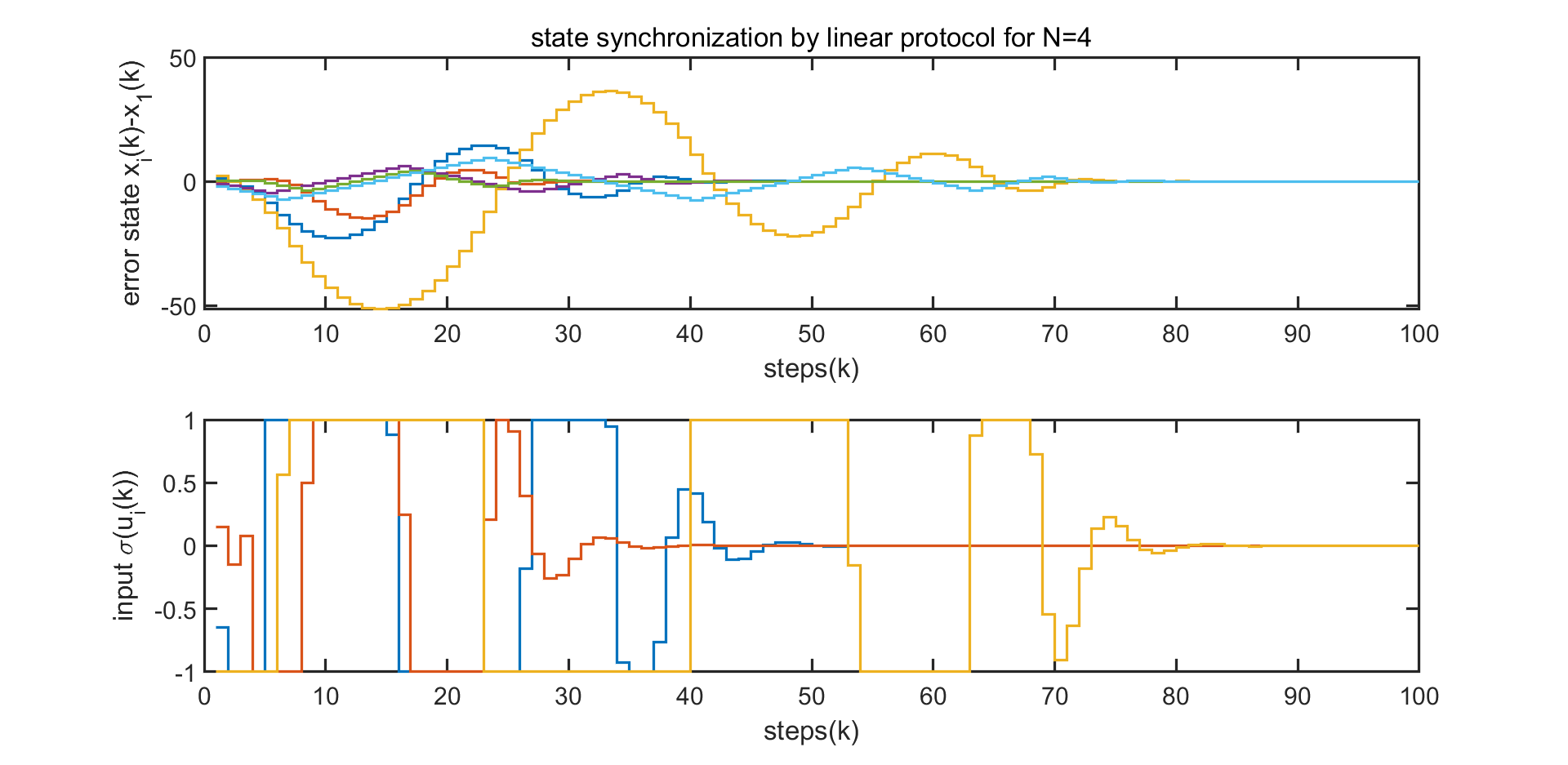} \centering
	\caption{Global state synchronization for
		MAS with communication graph
		I when $k_1=0.5$ and $k_2=1$.}\label{results_case1.1}
\end{figure}

\subsection*{Case II: 7-agent graph}
Then, we consider a MAS with $7$ agents
$N=7$. The associated adjacency matrix to the communication network is assumed to be $\mathcal{A}_{II}$ where $a_{21}=a_{32}=a_{43}=a_{24}=a_{54}=a_{47}=a_{65}=a_{76}=1$.
%\begin{figure}[ht!]
%	\includegraphics[width=5.5cm, height=2.7cm]{Ngraph_2}
%	\centering
%	\caption{Communication graph for case I}\label{graph_2}
%\end{figure}

In this case, we still use protocol \eqref{protoclsim}. The simulation result are demonstrated in Figure
\ref{results_case2.1}. 

%The results also show that the
%protocol design is effective for different ($k_1, k_2$) satisfying condition \eqref{cond1}.

\begin{figure}[ht!]
	\includegraphics[width=8.5cm, height=5cm]{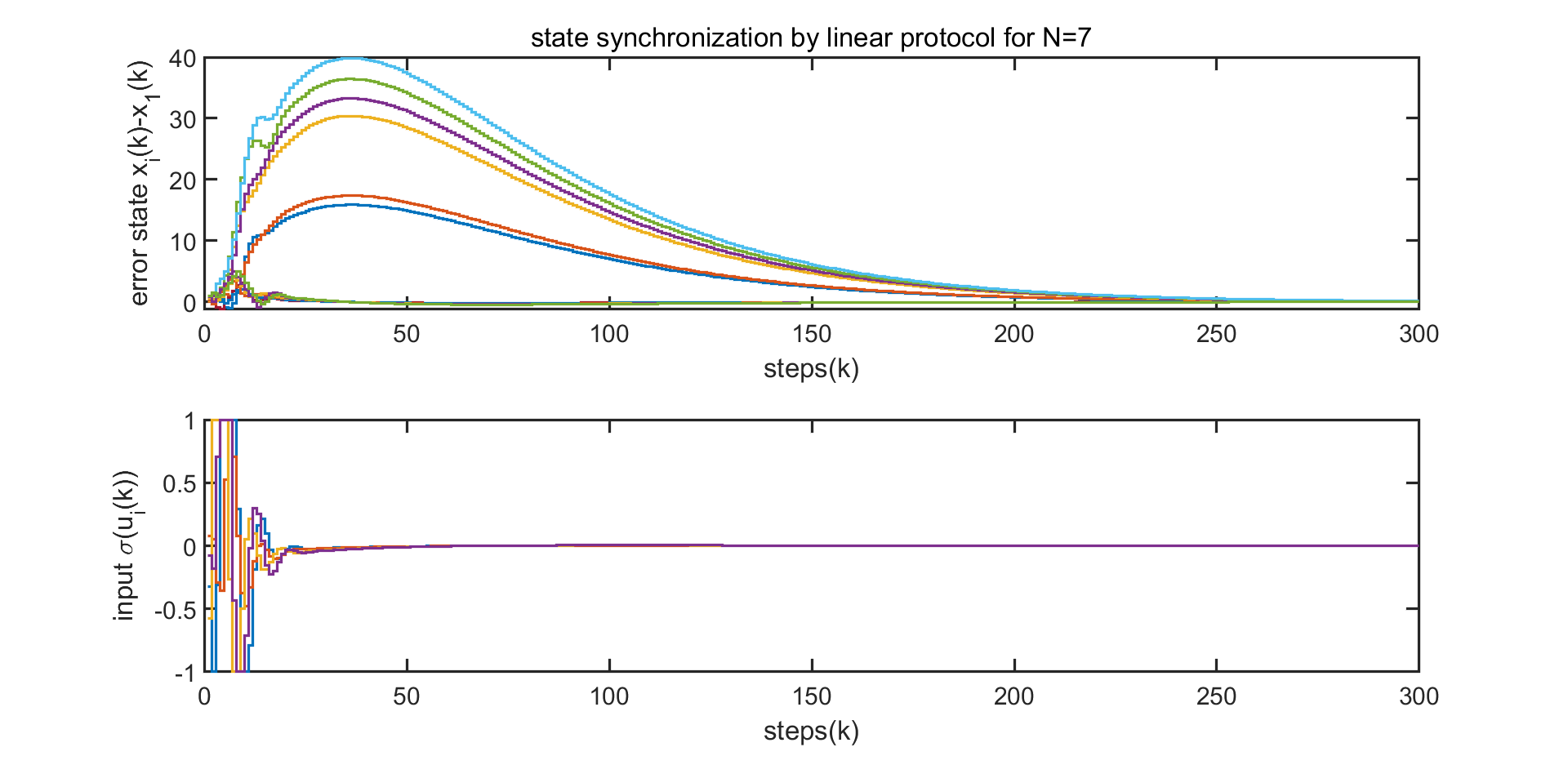} \centering
	\caption{Global state synchronization for
		MAS with communication graph
		II when $k_1=0.5$ and $k_2=1$.}\label{results_case2.1}
\end{figure}

\subsection*{Case III: 60-agent graph}
Finally, we consider a MAS with $60$ agents $N=60$ 
and a directed loop graph, where the associated adjacency matrix is assumed to be $\mathcal{A}_{III}$ only with $a_{i+1,i}=a_{1,60}=1$ and $i=1,\cdots,59$.

%illustrated in Figure \ref{graph_4}.
%\begin{figure}[ht!]
%	\includegraphics[width=5.5cm, height=5.5cm]{graph60}
%	\centering
%	\caption{Communication graph for case III}\label{graph_4}
%\end{figure}

The simulation result by using protocol \eqref{protoclsim} is demonstrated in Figure
\ref{results_case3.1}. 
\begin{figure}[ht!]
	\includegraphics[width=8.5cm, height=5cm]{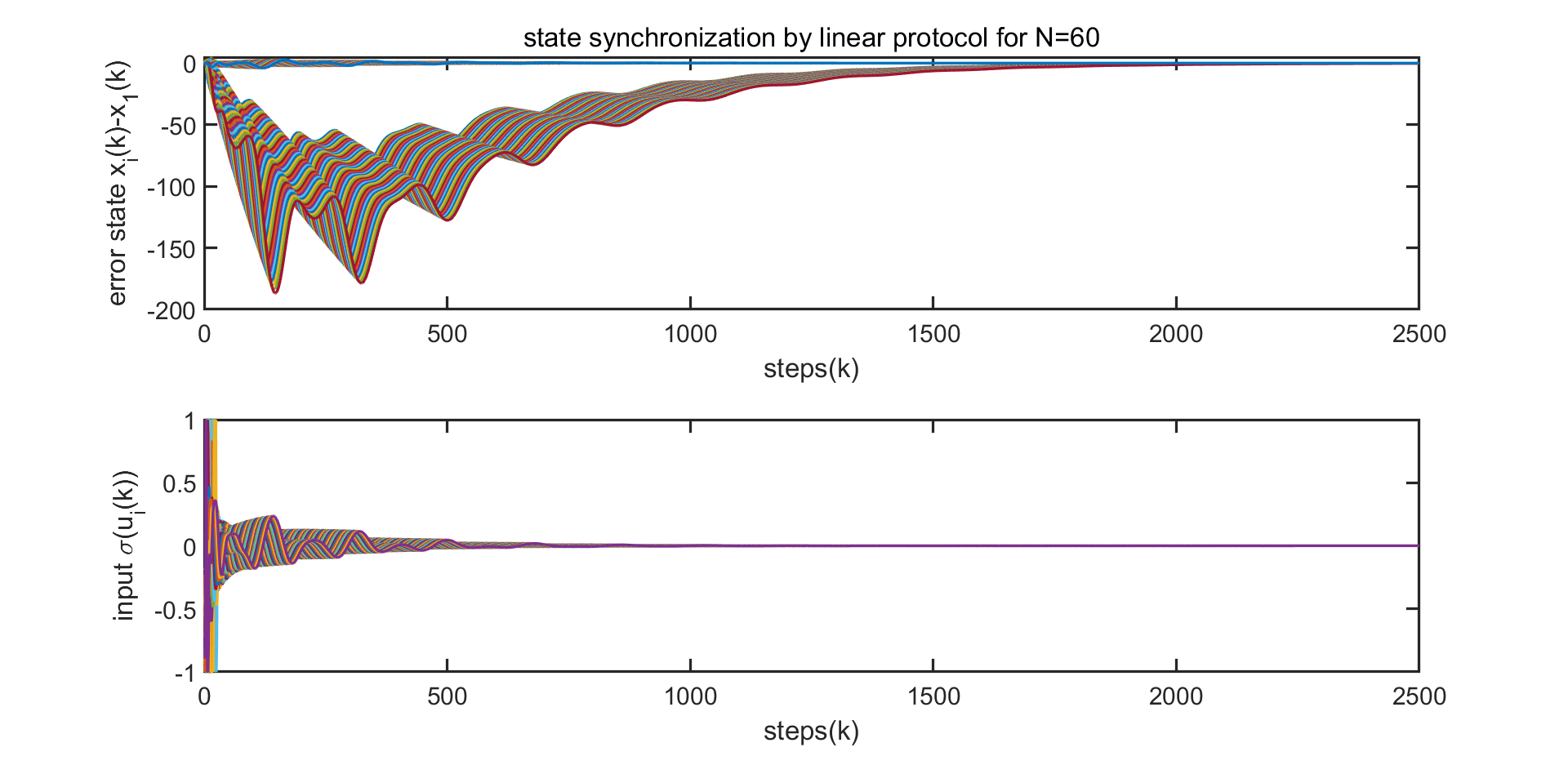} \centering
	\caption{Global state synchronization for
		MAS with communication graph
		III when $k_1=0.5$ and $k_2=1$.}\label{results_case3.1}
\end{figure}

All above simulation results with different graphs show that the
protocol design is independent of the communication graph and is scale
free so that we can achieve synchronization with one-shot protocol
design, for any graph with any number of agents.

\section{Conclusion}

In this paper, we have developed a scale-free linear protocol design to
achieve global state synchronization for discrete-time
double-integrator MAS subject to actuator saturation. The scale-free
protocols are designed solely based on agent models without utilizing
any information and are universal, it means that for any number of
agents and any communication graph. Meanwhile, we provide a solvable zone to protocol's parameter, and all parameter pair in this zone can achieve state synchronization of discrete-time double integrator MAS with input saturation. 

\bibliographystyle{plain}
\bibliography{referenc}

\begin{thebibliography}{10}

\bibitem{bai-arcak-wen}
H.~Bai, M.~Arcak, and J.~Wen.
\newblock {\em Cooperative control design: a systematic, passivity-based
  approach}.
\newblock Communications and Control Engineering. Springer Verlag, 2011.

\bibitem{bullobook}
F.~Bullo.
\newblock {\em Lectures on network systems}.
\newblock Kindle Direct Publishing, 2019.

\bibitem{chopra-tac}
N.~Chopra.
\newblock Output synchronization on strongly connected graphs.
\newblock {\em IEEE Trans. Aut. Contr.}, 57(1):2896--2901, 2012.

\bibitem{chowdhury-khalil}
D.~Chowdhury and H.~K. Khalil.
\newblock Synchronization in networks of identical linear systems with reduced
  information.
\newblock In {\em American Control Conference}, pages 5706--5711, Milwaukee,
  WI, 2018.

\bibitem{chu-yuan-zhang}
H.~Chu, J.~Yuan, and W.~Zhang.
\newblock Observer-based consensus tracking for linear multi-agent systems with
  input saturation.
\newblock {\em IET Control Theory and Applications}, 9(14):2124--2131, 2015.

\bibitem{eichler-werner}
A.~Eichler and H.~Werner.
\newblock Closed-form solution for optimal convergence speed of multi-agent
  systems with discrete-time double-integrator dynamics for fixed weight
  ratios.
\newblock {\em Syst. \& Contr. Letters}, 71:7--13, 2014.

\bibitem{royle-godsil}
C.~Godsil and G.~Royle.
\newblock {\em Algebraic graph theory}, volume 207 of {\em Graduate Texts in
  Mathematics}.
\newblock Springer-Verlag, New York, 2001.

\bibitem{grip-yang-saberi-stoorvogel-automatica}
H.F. Grip, T.~Yang, A.~Saberi, and A.A. Stoorvogel.
\newblock Output synchronization for heterogeneous networks of
  non-introspective agents.
\newblock {\em Automatica}, 48(10):2444--2453, 2012.

\bibitem{hadjicostis-charalambous}
C.N. Hadjicostis and T.~Charalambous.
\newblock Average consensus in the presence of delays in directed graph
  topologies.
\newblock {\em IEEE Trans. Aut. Contr.}, 59(3):763--768, 2014.

\bibitem{kim-shim-back-seo}
H.~Kim, H.~Shim, J.~Back, and J.~Seo.
\newblock Consensus of output-coupled linear multi-agent systems under fast
  switching network: averaging approach.
\newblock {\em Automatica}, 49(1):267--272, 2013.

\bibitem{kocarev-book}
L.~Kocarev.
\newblock {\em Consensus and synchronization in complex networks}.
\newblock Springer, Berlin, 2013.

\bibitem{li-zhang}
T.~Li and J.~Zhang.
\newblock Consensus conditions of multi-agent systems with time-varying
  topologies and stochastic communication noises.
\newblock {\em IEEE Trans. Aut. Contr.}, 55(9):2043--2057, 2010.

\bibitem{li-xiang-wei}
Y.~Li, J.~Xiang, and W.~Wei.
\newblock Consensus problems for linear time-invariant multi-agent systems with
  saturation constraints.
\newblock {\em IET Control Theory and Applications}, 5(6):823--829, 2011.

\bibitem{li-duan-chen}
Z.~Li, Z.~Duan, and G.~Chen.
\newblock Consensus of discrete-time linear multi-agent system with
  observer-type protocols.
\newblock {\em Discrete and Continuous Dynamical Systems. Series B},
  16(2):489--505, 2011.

\bibitem{liu-saberi-stoorvogel-IJRNC-2022}
Z.~Liu, A.~Saberi, and A.~A. Stoorvogel.
\newblock Scale-free collaborative protocols for global regulated state
  synchronization of discrete-time homogeneous networks of non-introspective
  agents in presence of input saturation.
\newblock {\em Int. J. Robust \& Nonlinear Control}, 2022.
\newblock Early access, DOI: 10.1002/rnc.6087.

\bibitem{liu-saberi-stoorvogel-zhang-ijrnc}
Z.~Liu, A.~Saberi, A.~A. Stoorvogel, and M.~Zhang.
\newblock Passivity-based state synchronization of homogeneous multiagent
  systems via static protocol in the presence of input saturation.
\newblock {\em Int. J. Robust \& Nonlinear Control}, 28(7):2720--2741, 2018.

\bibitem{liu2018regulated}
Z.~Liu, A.~Saberi, A.A. Stoorvogel, and D.~Nojavanzadeh.
\newblock Regulated state synchronization of homogeneous discrete-time
  multi-agent systems via partial state coupling in presence of unknown
  communication delays.
\newblock {\em IEEE Access}, 7:7021--7031, 2019.

\bibitem{liu-saberi-stoorvogel-donya-inputsaturation-automatica}
Z.~Liu, A.~Saberi, A.A. Stoorvogel, and D.~Nojavanzadeh.
\newblock Global regulated state synchronization for homogeneous networks of
  non-introspective agents in presence of input saturation: Scale-free
  nonlinear and linear protocol designs.
\newblock {\em Automatica}, 119:109041(1--8), 2020.

\bibitem{liu-zhang-saberi-stoorvogel-auto}
Z.~Liu, M.~Zhang, A.~Saberi, and A.~A. Stoorvogel.
\newblock State synchronization of multi-agent systems via static or adaptive
  nonlinear dynamic protocols.
\newblock {\em Automatica}, 95:316--327, 2018.

\bibitem{liu-zhang-saberi-stoorvogel-ejc}
Z.~Liu, M.~Zhang, A.~Saberi, and A.A. Stoorvogel.
\newblock Passivity based state synchronization of homogeneous discrete-time
  multi-agent systems via static protocol in the presence of input delay.
\newblock {\em European Journal of Control}, 41:16--24, 2018.

\bibitem{meng-zhao-lin-2013}
Z.~Meng, Z.~Zhao, and Z.~Lin.
\newblock On global leader-following consensus of identical linear dynamic
  systems subject to actuator saturation.
\newblock {\em Syst. \& Contr. Letters}, 62(2):132--142, 2013.

\bibitem{mesbahi-egerstedt}
M.~Mesbahi and M.~Egerstedt.
\newblock {\em Graph theoretic methods in multiagent networks}.
\newblock Princeton University Press, Princeton, 2010.

\bibitem{saber-murray2}
R.~Olfati-Saber and R.M. Murray.
\newblock Consensus problems in networks of agents with switching topology and
  time-delays.
\newblock {\em IEEE Trans. Aut. Contr.}, 49(9):1520--1533, 2004.

\bibitem{ren}
W.~Ren.
\newblock On consensus algorithms for double-integrator dynamics.
\newblock {\em IEEE Trans. Aut. Contr.}, 53(6):1503--1509, 2008.

\bibitem{ren-beard}
W.~Ren and R.W. Beard.
\newblock Consensus seeking in multiagent systems under dynamically changing
  interaction topologies.
\newblock {\em IEEE Trans. Aut. Contr.}, 50(5):655--661, 2005.

\bibitem{ren-book}
W.~Ren and Y.C. Cao.
\newblock {\em Distributed coordination of multi-agent networks}.
\newblock Communications and Control Engineering. Springer-Verlag, London,
  2011.

\bibitem{saberi-stoorvogel-zhang-sannuti}
A.~Saberi, A.~A. Stoorvogel, M.~Zhang, and P.~Sannuti.
\newblock {\em Synchronization of multi-agent systems in the presence of
  disturbances and delays}.
\newblock Birkhäuser, Cham, 2022.

\bibitem{scardovi-sepulchre}
L.~Scardovi and R.~Sepulchre.
\newblock Synchronization in networks of identical linear systems.
\newblock {\em Automatica}, 45(11):2557--2562, 2009.

\bibitem{seo-shim-back}
J.H. Seo, H.~Shim, and J.~Back.
\newblock Consensus of high-order linear systems using dynamic output feedback
  compensator: low gain approach.
\newblock {\em Automatica}, 45(11):2659--2664, 2009.

\bibitem{su-chen}
H.~Su and M.Z.Q. Chen.
\newblock Multi-agent containment control with input saturation on switching
  topologies.
\newblock {\em IET Control Theory and Applications}, 9(3):399--409, 2015.

\bibitem{su-chen-lam-lin}
H.~Su, M.Z.Q. Chen, J.~Lam, and Z.~Lin.
\newblock Semi-global leader-following consensus of linear multi-agent systems
  with input saturation via low gain feedback.
\newblock {\em IEEE Trans. Circ. \& Syst.-I Regular papers}, 60(7):1881--1889,
  2013.

\bibitem{su-huang-tac}
Y.~Su and J.~Huang.
\newblock Stability of a class of linear switching systems with applications to
  two consensus problem.
\newblock {\em IEEE Trans. Aut. Contr.}, 57(6):1420--1430, 2012.

\bibitem{tuna1}
S.E. Tuna.
\newblock {LQR}-based coupling gain for synchronization of linear systems.
\newblock Available: arXiv:0801.3390v1, 2008.

\bibitem{tuna2}
S.E. Tuna.
\newblock Synchronizing linear systems via partial-state coupling.
\newblock {\em Automatica}, 44(8):2179--2184, 2008.

\bibitem{tuna3}
S.E. Tuna.
\newblock Conditions for synchronizability in arrays of coupled linear systems.
\newblock {\em IEEE Trans. Aut. Contr.}, 55(10):2416--2420, 2009.

\bibitem{wang-saberi-stoorvogel-grip-yang}
X.~Wang, A.~Saberi, A.A. Stoorvogel, H.F. Grip, and T.~Yang.
\newblock Synchronization in a network of identical discrete-time agents with
  uniform constant communication delay.
\newblock {\em Int. J. Robust \& Nonlinear Control}, 24(18):3076--3091, 2014.

\bibitem{wieland-kim-allgower}
P.~Wieland, J.S. Kim, and F.~Allg\"ower.
\newblock On topology and dynamics of consensus among linear high-order agents.
\newblock {\em International Journal of Systems Science}, 42(10):1831--1842,
  2011.

\bibitem{wu-book}
C.W. Wu.
\newblock {\em Synchronization in complex networks of nonlinear dynamical
  systems}.
\newblock World Scientific Publishing Company, Singapore, 2007.

\bibitem{yang-meng-dimarogonas-johansson}
T.~Yang, Z.~Meng, D.V. Dimarogonas, and K.H. Johansson.
\newblock Global consensus for discrete-time multi-agent systems with input
  saturation constraints.
\newblock {\em Automatica}, 50(2):499--506, 2014.

\bibitem{you-xie}
K.~You and L.~Xie.
\newblock Network topology and communication data rate for consensusability of
  discrete-time multi-agent systems.
\newblock {\em IEEE Trans. Aut. Contr.}, 56(10):2262--2275, 2011.

\bibitem{zhang-chen-su}
L.~Zhang, M.Z.Q. Chen, and H.~Su.
\newblock Observer-based semi-global consensus of discrete-time multi-agent
  systems with input saturation.
\newblock {\em Transactions of the Institute of Measurement and Control},
  38(6):665--674, 2016.

\bibitem{zhang-saberi-stoorvogel-continues-discrete}
M.~Zhang, A.~Saberi, and A.A. Stoorvogel.
\newblock Synchronization in a network of identical continuous-or discrete-time
  agents with unknown nonuniform constant input delay.
\newblock {\em Int. J. Robust \& Nonlinear Control}, 28(13):3959--3973, 2018.

\end{thebibliography}
\end{document}